\documentclass[12pt]{article}
\usepackage{amsthm}
\usepackage{latexsym}
\pdfoutput=1
\usepackage[latin1]{inputenc}
\usepackage{amsfonts}
\usepackage{natbib}
\usepackage{graphics}
\usepackage{graphicx}
\usepackage{multirow}
\usepackage{hhline}
\usepackage{amsmath}
\usepackage{amsfonts}
\usepackage{amssymb}

\def\E{\operatorname{\bf E}}

\def\Q{\operatorname{\bf Q}}
\def\1{{\bf 1}}
\def\R{\mathbb R}

\def\Q{\mathbb Q}

\def\beq{\begin{equation}}
\def\eeq{\end{equation}}

\def\process#1{#1=\{#1_t\}_{t\geq 0}}
\newtheorem{theorem}{Theorem}[section]

\newtheorem{corollary}{Corollary}[section]

\begin{document}
\title{Barrier Options under Lévy Processes: a Short-Cut
}
\author{
Jos\'{e} Fajardo\thanks{Brazilian School of Public and Business
Administration, Getulio Vargas Foundation, Praia de Botafogo 190,
22253 900 - Rio de Janeiro, RJ, Brazil. E-mail address:
jose.fajardo@fgv.br} }
\date{\today } \maketitle
\begin{abstract}
In this paper we present a simple way  to price a class of barrier
options when the underlying process is driven by a huge class of
Lévy processes. To achieve our goal we assume that our market
satisfies a symmetry property. In case of not satisfying that
property some approximations and relationships can be
obtained.\vspace{3 mm }
\end{abstract}
\section{Introduction}
Recently, the relationship between the implied volatility symmetry
and other symmetry concepts has been established, as
\cite{FajardoMordecki2006b} have shown for L\'evy process and
\cite{CarrLee2008}  have shown for local/stochastic volatility
models and time-changed Lévy processes. Also,
\cite{FajardoMordecki2008} have studied the relationship of symmetry
properties with the skewness premium. Moreover, it is known that
important applications, such as the construction of semi-static
hedges for exotic options can be
obtained, as \cite{CEG98} and \cite{CarrLee2008} have shown, and
its extension to multivariate derivatives due to \cite{molchanov}.\\

The importance of such symmetry properties have demanded the
analysis of conditions to verify what kind of underlying processes
satisfy such properties, as \cite{CarrLee2008},
\cite{Fajardomordecki2005}, \cite{FajardoMordecki2007} and \cite{MT} results have shown.\\

On the other hand, \cite{Bates97} verified that put-call symmetry
does not hold in practice, by constructing a hypothesis test that
compares the observed relative call-put prices and the ratio given
by Bates' rule. The absence of symmetry has also been reported by
\cite{CarrWu07}. They found asymmetric implied volatility smiles in
currency options. But, besides these empirical findings other
applications of symmetry properties can justify their use, as for
example the pricing of credit instruments, such as barrier
contracts, as \cite{CarrLee2008} have shown by constructing
semi-static hedging for a class of barrier options.\\

 In the context of Lévy processes the pricing of barrier contracts is a delicated issue. We find
 some contributions, as \cite{lipton2002} who
considers jump-diffusion process with exponentially distributed
jumps. Also, \cite{KouWang2004} price barrier option with the double
exponential jump-diffussion model. \cite{asmussenetal07} consider
the CGMY model to price some kind of barrier options. For more
general Lévy processes we have the contributions of \cite{rogers},
for spectrally one-sided Lévy processes, and \cite{asmussenetal04},
for phase-type jumps. More recently, Fourier transform methods have
been used to price barrier options under more general Lévy processes
see \cite{egp09}, \cite{egp11} and \cite{cc10}. For more references see \cite{wimcariboni}.\\

In this paper we present a very simple way to price special kinds of
exotic options, such as digital call and put options,
asset-or-nothing options, under a huge family of Lévy processes
assuming that a symmetry property holds. When symmetry property does
not hold we obtain some approximations and we price the down-and-in
power option, by extending Th. 7.6 in \cite{CarrLee2008} for our
class of Lévy processes.\\

 The paper is organized as follows, in Section 2 we introduce our model. In Section 3 we present our main
results for symmetric markets. In Section 4 we present the results
under absence of symmetry. In Section 5 we present some numerical
examples and the last section concludes.
\section{Market Model}
Consider a real valued stochastic process $\process X$, defined on a
stochastic basis ${\cal B}=(\Omega, {\cal F},{\bf F}=({\cal
F}_t)_{t\geq 0}, \Q)$, being c\`adl\`ag, adapted, satisfying
$X_0=0$, and such that for $0\leq s< t$ the random variable
$X_t-X_s$ is independent of the $\sigma$-field ${\cal F}_s$, with a
distribution that only depends on the difference $t-s$. Assume also
that the stochastic basis ${\cal B}$ satisfies the usual conditions
(see \cite{jacodShiryaev87}). The process $X$ is a L\'evy process,
and is also called a process with stationary independent increments. For L\'evy process in Finance see \cite{Schoutens2003} and \cite{CT04}.\\

In order to characterize the law of $X$ under $\Q$, consider, for
$q\in\R$ the L\'evy-Khinchine formula, that states
\begin{equation}\label{e:lk}
\E e^{iq X_t}= \exp\Big\{t\Big[iaq-\frac 12\sigma^2q^2+
\int_{\R}\big(e^{iq y}-1-iq h(y)\big)\Pi(dy)\Big]\Big\},
\end{equation}
with
\begin{equation*}
h(y)=y{\bf 1}_{\{|y|<1\}}
\end{equation*}
a fixed truncation function, $a$ and $\sigma\geq 0$ real constants,
and $\Pi$ a positive measure on
${\R}\setminus\{0\}$\footnote{$\Pi(\{0\})$ could be defined as 0.
Here we follows \cite{CT04}.} such that $\int (1\wedge
y^2)\Pi(dy)<+\infty$, called the \emph{L\'evy measure}. The triplet
$(a,\sigma^2,\Pi)$ is the \emph{characteristic triplet} of the
process, and completely
determines its law.\\

Now we use the extension to the complex plane used by
\cite{lewis01}, that is we define the Lévy-Khinchine formula in the
strip $\{z:a<Im(z)<b\}$ where $a\leq -1$ and $b\geq 0$. Then we can
define the {\it characteristic exponent} of the process $X$, in this
strip, by:

\begin{equation}\label{e:char-exp}
\psi(z)=azi-\frac 12\sigma^2z^2+
\int_{\R}\big(e^{izy}-1-izh(y)\big)\Pi(dy)
\end{equation}
this function $\psi$ is also called the {\it cumulant} of $X$,
having $\E|e^{zX_t}|<\infty$ for all $t\geq 0$, and $\E
e^{zX_t}=e^{t\psi(z)}$. For $t=1$, Formula \eqref{e:char-exp}
reduces to exponent of eq. \eqref{e:lk} with $Im(z)=0$.
\subsection{L\'evy market} By a \emph{L\'evy market} we mean a
model of a financial market with two assets: a deterministic savings
account $\process B$, with
\begin{equation*}
B_t=e^{rt},\qquad r\ge 0,
\end{equation*}
where we take $B_0=1$ for simplicity, and a stock $\process S$, with
random evolution modeled by
\begin{equation}\label{e:market}
S_t=S_0e^{X_t},\qquad S_0=e^x>0,
\end{equation}
where $\process X$ is a L\'evy process.\\

In this model we assume that the stock pays dividends, with constant
rate $\delta\ge 0$, and that the given probability measure $\Q$ is
the chosen equivalent martingale measure. In other words, prices are
computed as expectations with respect to
$\Q$, and the discounted and reinvested process $\{e^{-(r-\delta)t}S_t\}$ is a $\Q$-martingale.\\

In terms of the characteristic exponent of the process this means
that

\begin{equation}\label{e:psiuno}
\psi(1)=r-\delta,
\end{equation}
based on the fact, that $\E
e^{-(r-\delta)t+X_t}=e^{-t(r-\delta-\psi(1))}=1$, and condition
\eqref{e:psiuno} can also be formulated in terms of the
characteristic triplet of the process $X$ as
\begin{equation}\label{e:a}
a=r-\delta-\sigma^2/2-\int_{\R}\big(e^y-1-\mathbf{1}_{\{|y|<1\}}\big)\Pi(dy).
\end{equation}
Then,
\begin{equation}\label{ok}
\psi(z)=iz(r-\delta-\frac {\sigma^2}{2})-z^2\frac
{\sigma^2}{2}+\int_{-\infty}^{+\infty}[iz(1-e^y)+(e^{izy}-1)]\Pi(dy)
\end{equation}
Henceforth we will denote this exponent by $\psi_\beta$ due to its
future dependence on parameter $\beta$ of our jump structure.
\subsection{Market Symmetry}
Here we use the symmetry concept introduced in
\cite{FajardoMordecki2006b}. We define a L\'evy market to be
\emph{symmetric} when the following relation holds:

\begin{equation}\label{e:symmetry}
{\cal L}\big(e^{-(r-\delta)t+X_t}\mid \Q\big)={\cal
L}\big(e^{-(\delta-r)t-X_t}\mid\tilde{\Q}\big),
\end{equation}
meaning equality in law. Otherwise we call the Lévy market
\emph{asymmetric}. As \cite{FajardoMordecki2006b} pointed out, a
necessary and sufficient condition for \eqref{e:symmetry} to hold is

\begin{equation}\label{e:pisymmetric}
\Pi(dy)=e^{-y}\Pi(-dy).
\end{equation}
This ensures $\tilde{\Pi}=\Pi$, and from this follows
$a-(r-\delta)=\tilde{a}-(\delta-r)$, giving \eqref{e:symmetry}, as
we always have $\tilde{\sigma}=\sigma$.\\

Moreover, in L\'evy markets with jump measure of the form

\begin{equation}\label{e:beta}
\Pi(dy)=e^{\beta y}\Pi_0(dy),
\end{equation}
where $\Pi_0(dy)$ is a symmetric measure, i.e.
$\Pi_0(dy)=\Pi_0(-dy)$ and $\beta$ is a parameter that describe the
asymmetry of the jumps, everything with respect to the risk
neutral measure $\Q$.\\

 As a consequence of \eqref{e:pisymmetric},
\cite{FajardoMordecki2006b} found that the market is
symmetric if and only if $\beta=-1/2$.\\

Now substituting in eq. (\ref{ok}), we observe the dependence of the
characteristic exponent on the parameter $\beta$:
\begin{equation}
\psi_{\beta}(z)=iz(r-\delta-\frac {\sigma^2}{2})-z^2\frac
{\sigma^2}{2}+\int_{-\infty}^{+\infty}[iz(1-e^y)+(e^{izy}-1)]e^{\beta
y}\Pi_0(dy), \label{fundamental}
\end{equation}
from her we obtain the partial derivative with respect to $\beta$:
\begin{equation}
\frac {\partial {\psi_{\beta}}}{\partial
\beta}(-z)=\int_{-\infty}^{+\infty}[(e^{-izy}-1)-iz(1-e^y)]ye^{\beta
y}\Pi_0(dy).\label{alan10}
\end{equation}
It will be needed in the approximations to be presented later.
\section{Results}
Let $\sigma_{imp}(x)(x,\beta)$ denote the Black-Scholes implied
volatility, that depends on the log-moneyness $x=\log K/F$, where
$K$ is the strike price and $F$ the forward price, the symmetry
parameter $\beta$ and denote  by $f_x$ the price of a European style
digital call option with maturity $T$ and barrier $K_x$, i.e. at
Maturity derivative pays off $f_x(y)=1_{\{e^y>K_x\}}$. The put
version is denoted by $g_x$ and pay off given by
$g_x(y)=1_{\{e^y\leq K_x\}}$. As in \cite{egp09}, we can take
$y=X_T+ln(S_0)+(r-\delta)T$, then $f_0(X_T)=1_{\{e^{X_T}>1\}}$. Now
our main result.
\begin{theorem}{\label{main}}
When symmetry holds ($\beta = -1/2$) the price of the digital call
option with barrier $K_0$ is given by:

$$f_0=e^{-rT}N(-\frac {\sigma\sqrt{T}}{2}).$$
In the case of an European digital put option, the price is given by:
$$g_0= e^{-rT}\left[1-N(-\frac {\sigma\sqrt{T}}{2})\right]$$

\end{theorem}
\begin{proof}
Lets denote $BS_c(\sigma_{imp}(x,\beta))$ the Black and Scholes
price of and European call with maturity $T$, strike $K$ and spot
price $S_0$. When the risk neutral interest rate is denoted by $r$
and the price is driven by a L\'{e}vy process with characteristic
triplet $(a,\sigma,\Pi)$, that BS-price is given by:
$$BS_c(\sigma_{imp}(x,\beta))=\sf{E}e^{-rT}\left(S_0e^{X_T}-K\right)^+.$$
Now, we know by L\'{e}vy-Khintchine formula
$$
\psi_{\beta}(z)=iza-z^2\frac
{\sigma^2}{2}+\int_{-\infty}^{+\infty}(e^{izy}-1)e^{\beta
y}\Pi_0(dy),
$$
when $\Pi(dy)=e^{\beta y}\Pi_0(dy)$. Also, as we are interested in
properties of option prices we need that discounted prices process
be martingales, then as showed in (\ref{ok}), characteristic
exponent must be given by:
\begin{equation}
\psi_{\beta}(z)=iz(r-\delta-\frac {\sigma^2}{2})-z^2\frac
{\sigma^2}{2}+\int_{-\infty}^{+\infty}[iz(1-e^y)+(e^{izy}-1)]e^{\beta
y}\Pi_0(dy). \label{fundamental}
\end{equation}

Then, we can compute the BS-price as:
\begin{eqnarray}
  BS_c(\sigma_{imp}(x,\beta)) &=&  \int_{-\infty}^{+\infty}e^{-rT}(S_0e^x-K)^+f_{X_T}(x)dx\\
   &=& \frac {e^{-rT}}{2\pi}\int\frac
   {-K^{iz+1}e^{-iz(\ln(S_0)+(r-q)T)}}{z(z-i)}e^{T\psi_{\beta}(-z)}dz
\end{eqnarray}
where the last equality was obtained using Parserval. By
\cite{lewis01} and \cite{lipton2002}, we know that the option value
($V$) of an European call with maturity $T$ and strike $K_x$ is
given by:
\begin{equation}
V(x,T)=\frac
{e^{-rT}}{2\pi}\int_{iv+\R}e^{-iz(\ln(S_0)+(r-q)T)}e^{T\psi_{\beta}(-z)}\frac{-K_x^{iz+1}}{z(z-i)}dz,\label{lewis}
\end{equation}
 with $v\in (a,b)$ and real constants $a<-1$ and $b>0$, where $K_x=S_0e^{(r-q)T}e^x$. Rewriting, we have
$$
V(x,T)=\frac
{-K_{x}e^{-rT}}{2\pi}\int_{iv+\R}e^{izx}\frac{e^{T\psi_{\beta}(-z)}}{z(z-i)}dz,\;\;v>1.
$$
Observe that \begin{eqnarray*}
 \frac {\partial
BS_c(K,\sigma_{imp}(x,\beta))}{\partial x}&=&\frac{\partial
BS_c(K,\sigma_{imp}(x,\beta))}{\partial K}\frac{\partial K}{\partial
x}+\frac {\partial BS_c(K,\sigma_{imp}(x,\beta))}{\partial
\sigma}\frac{\partial \sigma}{\partial x}\\
&=&-N(d_2)Ke^{-rT}+\frac{\partial
BS_c(K,\sigma_{imp}(x,\beta))}{\partial
\sigma}\frac{\partial \sigma}{\partial x}\\
&=&\frac {\partial V}{\partial x}(x,T).\\
\end{eqnarray*}
Where the last equality is obtained using \eqref{lewis}. Then,
$$\hbox{sgn}(\frac {\partial
{\sigma_{imp}}}{\partial x}(x,\beta))=\hbox{sgn}\left(\frac
{\partial V}{\partial x}(x,T)+N(d_2(x))Ke^{-rT} \right),$$ now
deriving \eqref{lewis} w.r.t $x$,we have:
$$
\frac {\partial V}{\partial x}(x,T)=\frac
{Ke^{-rT}}{2\pi}\int_{iv+\R}
   e^{izx}\frac{1}{iz}e^{T\psi_{\beta}(-z)}dz,\;\;v>1.$$
   Then,
$$\hbox{sgn}(\frac {\partial
{\sigma_{imp}}}{\partial x}(x,\beta))=\hbox{sgn}\left(\frac
{1}{2\pi}\int_{iv+\R}
   e^{izx}\frac{1}{iz}e^{T\psi_{\beta}(-z)}dz+N(d_2(x)) \right),
   \forall x,
$$

when $x=0$,
$$\hbox{sgn}(\frac {\partial
{\sigma_{imp}}}{\partial x}(0,\beta))=\hbox{sgn}\left(\frac
{1}{2\pi}\int_{iv+\R}
   \frac{1}{iz}e^{T\psi_{\beta}(-z)}dz+N(-\frac{\sigma\sqrt{T}}{2})
\right).
$$
Now we know by a result due to \cite{FajardoMordecki2006b} that in
symmetric L\'evy markets ($\beta=-0.5$), the implied volatility must
be symmetric w.r.t  to the log-moneyness i.e. $ \frac {\partial
{\sigma_{imp}}}{\partial x}(0,-0.5)=0$, Then
$$\frac
{1}{2\pi}\int_{iv+\R}
   \frac{1}{iz}e^{T\psi_{-0.5}(-z)}dz=-N(-\frac{\sigma\sqrt{T}}{2}).
$$
Finally, remember that the price $f_x$ of a digital call option with
barrier $K_x$, computed in our Lévy market can be expressed
as\footnote{See \cite{egp09}.}: \begin{equation} f_x=-\frac
{e^{-rT}}{2\pi}\int_{iv+\R}
   e^{izx}\frac{1}{iz}e^{T\psi_{\beta}(-z)}dz,
   \forall x,\;\;v>0,\label{ernest}
\end{equation}
From here the result follows and the put case is an immediate
consequence of $g_x=e^{-rT}-f_x$.
\end{proof}
Now under symmetry we can price other barrier options as asset-or-
nothing barrier options,  i.e. at maturity derivative pays off
$p(S_T)=S_T1_{\{S_T>K\}}$.
\begin{corollary}
Under symmetry the price $p$ of an asset-or-nothing option with
barrier $S_0^2$ is given by:
$$p=e^{-rT}S_0\left[1-N(-\frac
{\sigma\sqrt{T}}{2})\right].$$
\end{corollary}
\begin{proof}
Using Corollary 2.14 in \cite{CarrLee2008} with the martingale
$S_t=\exp(X_t+\ln(S_0)+(r-\delta)t)$ and $H=1$, we have:
$$E\left[S_01_{\{S_T<1\}}\right]=E\left[S_T1_{\{S_T\geq
S_0^2\}}\right],$$ and by Th. \ref{main}, we have
$$p=e^{-rT}E\left[S_T1_{\{S_T\geq S_0^2\}}\right]=S_0e^{-rT}\left[1-N(-\frac
{\sigma\sqrt{T}}{2})\right]$$
\end{proof}
\section{Absence of Symmetry}
\subsection{Approximations}
Now denote by $I(\beta,x)$ the integral defined in \eqref{ernest}.
Let $I_\beta$ and $I_x$ denote the partial derivatives of
$I(\beta,x)$ with respect to $\beta$ and $x$. Then some
approximations can be obtained.
\begin{corollary}
If we consider not too asymmetric markets ($|\beta+0.5|< \epsilon$)
and near at-the-money digital call options ($|x|< \varepsilon$).
\begin{itemize}
\item For a barrier $K_x$, we have
$$I(-0.5,x)\approx e^{-rT}N(-\frac
{\sigma\sqrt{T}}{2})+xI_x(-0.5,0)$$
\item For any symmetry parameter $\beta$, we have:
$$I(\beta,0)\approx e^{-rT}N(-\frac
{\sigma\sqrt{T}}{2})+(\beta+0.5)I_\beta(-0.5,0)$$
\item In general, we have:
$$I(\beta,x)\approx e^{-rT}N(-\frac
{\sigma\sqrt{T}}{2})+(\beta+0.5)I_\beta(-0.5,0)+xI_x(-0.5,0)$$

\end{itemize}
\end{corollary}
\begin{proof}
It follows directly from Taylor approximations and our Theorem 3.1.
\end{proof}
Similar approximations can be obtained for the digital put option, since $g_x=e^{-rT}-I(\beta,x)$.\\

Observe that in the approximations we need to compute:
\begin{equation}
I_\beta(-0.5,0)=\frac {Te^{-rT}}{2\pi}\int_{iv+\R}
   \frac{1}{iz}e^{T\psi_{-0.5}(-z)}\frac {\partial
{\psi_{-0.5}}}{\partial
\beta}(-z)dz,\;v>0\label{intf1}\end{equation} and
\begin{equation}
I_x(-0.5,0)=\frac {e^{-rT}}{2\pi}\int_{iv+\R}
  e^{T\psi_{-0.5}(-z)}dz,\;v>0\label{intf2}\end{equation} then we need a specific expression for the characteristic exponent
 and using FFT techniques we can compute the integrals, as we will
 present in Section 5.
\subsection{More General Results}
First, we need to extend Th. 7.6 in \cite{CarrLee2008}.
\begin{theorem}\label{peter2}
Let $X_t$ be a Lévy process with characteristic triplet
$(\mu,\sigma,\Pi)$ and characteristic exponent denoted by $\psi$,
with Lévy measure given by $\Pi(dy)=e^{\beta y}\Pi_0(dy)$, where
$\beta \neq -0.5$ (absence of symmetry) and define
$$
\alpha:=\left\{
  \begin{array}{ll}
    -\frac {\psi(2\beta i)}{2\beta}, & \beta\neq 0; \\
    \mu, & \beta=0.
  \end{array}
\right.
$$
Then for any payoff function\footnote{A payoff function is a
nonnegative Borel function on $\R$.} $f$, we have:
\begin{equation}
Ef(S_T)=E\left[\left(\frac {S_T}{S_0e^{\alpha
T}}\right)^{-2\beta}f\left(\frac {S_0^2e^{2\alpha
T}}{S_T}\right)\right]\label{peter}
\end{equation}
\end{theorem}
\begin{proof}
Here we follow \cite{CarrLee2008}, i.e, define the following
process:
$$Y_t:=-2\beta X_t+(T-t)\psi(2\beta i).$$
Denote by $\Pi_y$ its Lévy measure. Then,
$$\Pi_y(x)=e^{-(\beta+\frac 12)}\Pi_x(dx)=e^{-\frac 12}\Pi_0(dx),$$
as \cite{FajardoMordecki2006b} proved, process  $Y_t$ satisfy market
symmetry property. Moreover, as $E(e^{-2\beta X_t})=e^{t\psi(2\beta
i)}<\infty$, the price process
$\widetilde{S}_t=S_0^{-2\beta}e^{Y_t}=\left(S_te^{\alpha
(T-t)}\right)^{-2\beta}$ is a martingale and satisfy the desired symmetry property.\\

Finally, using our symmetric process $\widetilde{S}_t$, we can apply
Th. 2.5 in \cite{CarrLee2008} to the function $x\mapsto
f(x^{-1/2\beta})$, from here the result follows.
\end{proof}
\section{Numerical Examples}
Lets consider the Normal Inverse Gaussian(NIG) distribution. Which
has the following characteristic function:
\begin{equation}
\psi_\beta(z)=iz\mu+\delta\left(\sqrt{\alpha^2-\beta^2}-\sqrt{\alpha^2-(\beta+iz)^2}\right),\;-\beta-\alpha<Im(z)<\beta+\alpha.\label{intf3}\end{equation}
From here \begin{equation} \frac {\partial {\psi_{\beta}}}{\partial
\beta}(z)=\delta\left(\frac
{\beta+iz}{\sqrt{\alpha^2-(\beta+iz)^2}}-\frac{\beta}{\sqrt{\alpha^2-\beta^2}}\right)\label{intf4}\end{equation}
Before pricing Barrier options we need to estimate the parameters of
the NIG distribution. To this end we consider daily returns of
S\&P500 from 12/01/2009 to 12/01/2011 and using maximum likelihood
estimation we find
$(\mu,\alpha,\delta,\beta)=(0.0018,49.99,0.0085,-9.22)$. But we need
the risk-neutral parameters, then we use the density given by the
Esscher Transform. To compute this density we need the interest rate
so we use the interest rate given by the U.S. Treasury in that date
$r=0.0012$. Under this transformation we obtain the following
parameters:
$$(\mu^*,\alpha^*,\delta^*,\beta^*)=(0.0018,49.99,0.0085,-4.18)$$
With these parameters and the expressions \eqref{intf3} and
\eqref{intf4}, we can compute the integrals given in \eqref{intf1}
and \eqref{intf2} in the overlap of the strips giving $I_\beta(-0.5,0)=0.2621$ and $I_x(-0.5,0)=7.3212$ .\\

Now we want to price barrier options on S\&P500 that matures in one
year, i.e. $T=1$. We need the volatility then we use the VIX of
12/01/11, as a proxy of S\&P500 volatility, $\sigma=0.2741$.

\subsection{Call and put digital option}
Using our Th. \ref{main} we know that he price of a call and put
digital option ATM under symmetry is given by:

$$f_0=e^{-0.0012}N(-0.1371)=0.4449,$$
and $$g_0=e^{-0.0012}[1-N(-0.1371)]=0.5539$$ and the approximations
are given by:
\begin{itemize}
\item For any barrier $K_x$,  such $x\in [-\varepsilon,\varepsilon]$ we have
$$I(-0.5,x)\approx 0.4449+7.3212x$$
\item For any symmetry parameter $\beta\in [-0.5-\epsilon,-0.5+\epsilon]$, we have:
$$I(\beta,0)\approx 0.4449+0.2621(\beta+0.5)$$
\item In general, we have:
$$I(\beta,x)\approx 0.4449+0.2621(\beta+0.5)+7.3212x$$
\end{itemize}
Now using $\varepsilon=\epsilon=0.01$, we have the following figure.
\begin{figure}[!htb]
    \begin{center}
       \includegraphics[width=0.7\linewidth]{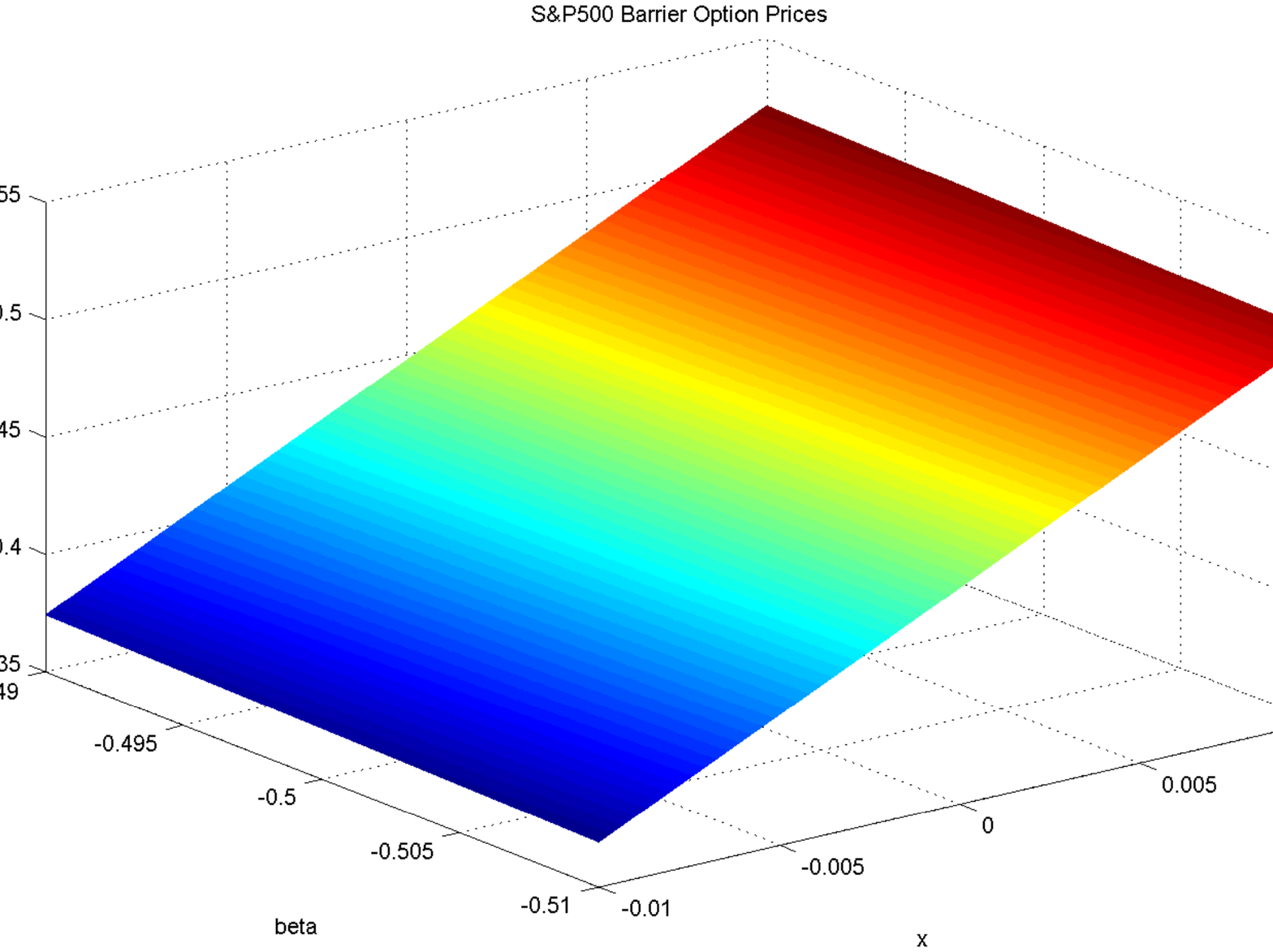}
            \caption{Barrier Option Prices for different Moneynesss and Symmetries}
        \label{fig:sixindex}
    \end{center}
\end{figure}
\subsection{Down-and-in power option}
Now to price a down-and-in power option on S\&P500 we observe that
in the NIG case we have: $$\alpha=\mu=0.0018.$$ Using $T=1$ and
applying Th. \ref{peter2} to the function $f(S_T)=1_{\{S_T>1\}}$, we
obtain:
$$E\left(1_{\{S_T>1\}}\right)=E\left(1_{\{S_T^{8.36}>1\}}\right)=E\left[\left(\frac {S_T}{S_0e^{0.0018}}\right)^{8.36}1_{\left\{\frac
{S_0^2e^{0.0036}}{S_T}>1\right\}}\right]$$

$$E\left(1_{\{\widetilde{S}_T>1\}}\right)=E\left[\left(\frac {(S^*_T)^{8.36}}{e^{0.015}}\right)1_{\left\{\frac
{1444.23e^{0.0036}}{S^*_T}>1\right\}}\right]$$ Where
$\widetilde{S}_T=S_T^{8.36}$ satisfies symmetry properties and
$S^*_T=e^{X_T}$ is an asymmetric process. Using Th. \ref{main}, we
obtain:

$$e^{-0.0012}E\left[\left(S^*_T\right)^{8.36} 1_{\left\{1449.4>S^*_T\right\}}\right]=0.4449*e^{0.015}=0.4516$$
\section{Conclusions}
We have obtained an easy way to compute the price some barrier
options under Lévy processes, using the symmetry property
obtained by \cite{FajardoMordecki2006b} and in case of absence of symmetry we have obtained some approximations and extending Th. 7.6 in \cite{CarrLee2008} we have obtained a general relationship.\\

Some possible extensions are of interest, as the pricing of the
class of Contingent Convertible with Cancelable Coupons (CoCa CoCos)
introduced by \cite{wimetall.} under Lévy processes.
\bibliographystyle{econometrica}
\bibliographystyle{econometrica}
\bibliography{catalog2010}
\end{document}